\documentclass[11pt, letterpaper]{amsart}

\usepackage{amsmath}
\usepackage{amssymb}
\usepackage{amsfonts}
\usepackage{ dsfont }
\usepackage{cases}
\usepackage{cite}
\usepackage[hyphens]{url}
\usepackage{hyperref}
\usepackage{amsthm}
\usepackage{upgreek}

\newtheorem{theorem}{Theorem}
\newtheorem{remark}{Remark}

\newcommand{\eps}{\varepsilon}

\title[]{On model-independent pricing/hedging using shortfall risk and quantiles}
\author[]{Erhan Bayraktar} \thanks{This research was supported in part by the National Science Foundation under grants DMS 0906257 and DMS 1118673.}  
\address{Department of Mathematics, University of Michigan}
\email{erhan@umich.edu}
\author[]{Zhou Zhou}
\address{Department of Mathematics, University of Michigan}
\email{zhouzhou@umich.edu}
\date{\today}
\begin{document}
\maketitle

\begin{abstract}
We consider the pricing and hedging of exotic options in a model-independent set-up using \emph{shortfall risk and quantiles}. We assume that the marginal distributions at certain times are given. This is tantamount to calibrating the model to call options with discrete set of maturities but a continuum of strikes. In the case of pricing with shortfall risk, we prove that the minimum initial amount is equal to the super-hedging price plus the inverse of the utility at the given shortfall level. In the second result, we show that the quantile hedging problem is equivalent to super-hedging problems for knockout options. These results generalize the duality results of \cite{MR1780323, MR1842286} to the model independent setting of \cite{Beiglbock1}.
\end{abstract}

\section{Set-up and the main results}
We will follow the setting in \cite{Beiglbock1}. Assume that in the market, there is a single risky asset  at discrete times $t=1,\dotso,n$. Let $S=(S_t)_{t=1}^n$ be the canonical process on the path space  $\mathbb{R}_+^n$, i.e, for  $(s_1,\dotso,s_n) \in \mathbb{R}_+^n$ we have that $S_i(s_1,\dotso,s_n)=s_i$. The random variable $S_i$ represents the price of the risky asset  at time $t=i$. We denote the current spot price of the asset as $S_0=s_0$. In addition, we assume that our model is calibrated to a continuum of call options with payoffs $(S_i-K)^+,\ K\in\mathbb{R}_+$ at each time $t=i$, and price 
$$\mathcal{C}(i,K)=\mathbb{E}^\mathbb{Q}\left[(S_i-K)^+\right].$$
It is well-known that knowing the marginal $S_i$ is equivalent to knowing the prices $\mathcal{C}(i,K)$ for all $K\geq0$; see \cite{BL78}. 
Hence, we will assume that the marginals of the stock price $S=(S_i)_{i=1}^n$ are given by $S_i\sim\mu_i$, where $\mu_1,\dotso,\mu_n$ are probability measures on $\mathbb{R}_+$. Let 
$$\mathcal{M}:=\{\mathbb{Q} \text { probability measure on } \mathbb{R}_+^n:\ S=(S_i)_{i=1}^n \text{ is a } \mathbb{Q}-\text{martingale};$$
$$\text{ for } i=1,\dotso,n,\ S_i \text{ has marginal } \mu_i \text{ and mean } s_0\}.$$
We make the standing assumption that $\mathcal{M}\neq\emptyset$.

Let us consider the semi-static trading strategies consisting of the sum of a static vanilla portfolio and a dynamic strategy in the stock We will by $\Delta$ the predictable process corresponding to the holdings on the stock. More precisely, the semi-static strategies generate payoffs of the form:
$$\Uppsi_{(u_i),(\Delta_j)}(s_1,\dotso,s_n)=\sum_{i=1}^n u_i(s_i)+\sum_{j=1}^{n-1}\Delta_j(s_1,\dotso,s_j)(s_{j+1}-s_j),\ s_1,\dotso,s_n\in\mathbb{R}_+,$$
where the functions $u_i:\ \mathbb{R}_+\rightarrow\mathbb{R}$ are $\mu_i$ integrable for $i=1,\dotso,n$, and the functions $\Delta_j:\ \mathbb{R}_+^j\rightarrow\mathbb{R}$ are assumed to be bounded measurable for $j=1,\dotso,n-1$. We denote $\Delta=(\Delta_1,\dotso,\Delta_{n-1})$.

Now we are ready to state our main results. We want to point out that similar results also hold for the continuous time version within the framework of \cite{Dolinsky1}, and their proofs are the same as  in the discrete time case presented here.
\begin{theorem}[Pricing using Shortfall Risk]
Let $\Phi: \mathbb{R}_+^n\rightarrow\mathbb{R}$ be an upper semi-continuous function such that
$$\Phi(s_1,\dotso,s_n)\leq K\cdot(1+s_1+\dotso+s_n)$$
for some constant $K$. Let $U$ be a nondecreasing concave function. Let $\alpha=U(\beta)$. If $U$ is strictly increasing around a neighborhood of $\beta$, then the following duality holds:
\begin{equation*}
\begin{split}
C:=&\inf\left\{\sum_{i=1}^n\mathbb{E}_{\mu_i}[u_i]: \exists\Delta\ \text{s.t. } \inf_{\mathbb{P}\in\mathcal{P}}\mathbb{E}^{\mathbb{P}}[U(\Uppsi_{(u_i),(\Delta_j)}-\Phi)]\geq\alpha\right\} \\
=&\sup_{\mathbb{Q}\in\mathcal{M}}\mathbb{E}^{\mathbb{Q}}\Phi+U^{-1}(\alpha)=:D. 
\end{split}
\end{equation*}
where $\mathcal{P}$ is any set of probability measures on $\mathbb{R}_+^n$ containing $\mathcal{M}$. Moreover, the supremum is attained.
\end{theorem}
\begin{proof}
Denote
$$A:=\left\{\sum_{i=1}^n\mathbb{E}_{\mu_i}[u_i]: \exists\Delta\ \text{s.t. } \inf_{\mathbb{P}\in\mathcal{P}}\mathbb{E}^{\mathbb{P}}[U(\Uppsi_{(u_i),(\Delta_j)}-\Phi)]\geq\alpha\right\}.$$
For any $\sum_{i=1}^n\mathbb{E}_{\mu_i}[u_i]\in A$, there exists $\Delta$, such that $\forall\mathbb{Q}\in\mathcal{M}$,
$$U\left(\sum_{i=1}^n\mathbb{E}_{\mu_i}[u_i]-\mathbb{E}^{\mathbb{Q}}\Phi\right)=U\left(\mathbb{E}^{\mathbb{Q}}\left[\Uppsi_{(u_i),(\Delta_j)}-\Phi\right]\right)\geq\mathbb{E}^{\mathbb{Q}}\left[U(\Uppsi_{(u_i),(\Delta_j)}-\Phi)\right]\geq\alpha,$$
due to Jensen's Inequality. Hence, $\sum_{i=1}^n\mathbb{E}_{\mu_i}[u_i]\geq\mathbb{E}^{\mathbb{Q}}\Phi+U^{-1}(\alpha), \ \forall\mathbb{Q}\in\mathcal{M}$, which implies $C\geq D$.

Applying Corollary 1.2 in \cite{Beiglbock1}, we know that the supremum in the definition of $D$ is attained and\begin{equation*}
\begin{split}
D=&\inf\left\{\sum_{i=1}^n\mathbb{E}_{\mu_i}[u_i]: \exists\Delta\ \text{s.t. } \Uppsi_{(u_i),(\Delta_j)}\geq\Phi+U^{-1}(\alpha)\right\}\\
=&\inf\left\{\sum_{i=1}^n\mathbb{E}_{\mu_i}[u_i]: \exists\Delta \text{ s.t. } U(\Uppsi_{(u_i),(\Delta_j)}-\Phi)\geq\alpha\right\}.
\end{split}
\end{equation*}
Denote the above set inside \lq\lq inf\rq\rq\ by $B$. Fix an arbitrary $\eps>0$, and let  $D+\eps>\sum_{i=1}^n\mathbb{E}_{\mu_i}[u_i]\in B$. Then there exists $\Delta$, such that $U(\Uppsi_{(u_i),(\Delta_j)}-\Phi)\geq\alpha$, which implies $\mathbb{E}^{\mathbb{P}}[U(\Uppsi_{(u_i),(\Delta_j)}-\Phi)]\geq\alpha$ for any probability measure $\mathbb{P}$. Therefore, $\sum_{i=1}^n\mathbb{E}_{\mu_i}[u_i]\in A$. Thus, $D+\eps>\sum_{i=1}^n\mathbb{E}_{\mu_i}[u_i]\geq C$. As a result we have that $D+\eps> C$ for all $\eps>0$, which implies $D\geq C$.
\end{proof}

\begin{remark}
This result is a generalization of Theorem 4.1 in  \cite{Nutz1} to the framework of \cite{Beiglbock1}. That is we allow the volatility to be uncertain, but restrict the set of probability measures by allowing the hedging strategies to use static trading in options.

\end{remark}

\begin{theorem}[Pricing by Quantiles]
Let $\Phi: \mathbb{R}_+^n\rightarrow\mathbb{R}$ be a continuous function such that
$$0\leq\Phi(s_1,\dotso,s_n)\leq K\cdot(1+s_1+\dotso+s_n),$$
for some constant $K$. Let $\mathcal{P}$ be any set of probability measures on $\mathbb{R}_+^n$ and $\alpha\in[0,1]$. Define
$$\mathcal{A}(\mathcal{P},\alpha):=\left\{H\in\mathcal{F}\text{ closed }:\ \inf_{\mathbb{P}\in\mathcal{P}}\mathbb{P}(H)\geq\alpha\right\}.$$
We require our semi-static hedging strategies $u_i,\ i=1\dotso,n$, and $\Delta_j,\ j=1,\dotso,n-1$, to be continuous functions. Then the following holds:
\begin{equation}\label{eq:quantile-hed}
\begin{split}
I:=&\inf\left\{\sum_{i=1}^n\mathbb{E}_{\mu_i}[u_i]: \exists\Delta,\ \text{s.t. } \Uppsi_{(u_i),(\Delta_j)}\geq0,\text{ and }\inf_{\mathbb{P}\in\mathcal{P}}\mathbb{P}\left(\Uppsi_{(u_i),(\Delta_j)}\geq\Phi\right)\geq\alpha\right\} \\
=&\inf_{H\in\mathcal{A}(\mathcal{P},\alpha)}\sup_{\mathbb{Q}\in\mathcal{M}}\mathbb{E}^{\mathbb{Q}}[\Phi 1_H]=:J\\
\end{split}
\end{equation}
\end{theorem}
\begin{proof}
For $H\in\mathcal{A}(\mathcal{P},\alpha)$, let $J(H):=\sup_{\mathbb{Q}\in\mathcal{M}}\mathbb{E}^{\mathbb{Q}}[\Phi 1_H]$. Since $H$ is closed and $\Phi$ is upper semi-continuous we can apply  in \cite[Corollary 1.2]{Beiglbock1} (and the explanation before this result where it is argued that it is sufficient to take $u_i$ and $\Delta_j$ to be continuous), to obtain
\begin{equation*}
\begin{split}
J(H)=&\inf\left\{\sum_{i=1}^n\mathbb{E}_{\mu_i}[u_i]: \exists\Delta,\ \text{s.t. } \Uppsi_{(u_i),(\Delta_j)}\geq\Phi 1_H\right\}\\
\geq&\inf\left\{\sum_{i=1}^n\mathbb{E}_{\mu_i}[u_i]: \exists\Delta,\ \text{s.t. } \Uppsi_{(u_i),(\Delta_j)}\geq0,\text{ and }\inf_{\mathbb{P}\in\mathcal{P}}\mathbb{P}\left(\Uppsi_{(u_i),(\Delta_j)}\geq\Phi\right)\geq\alpha\right\}=I,
\end{split}
\end{equation*}
which implies $J\geq I$.

For $\eps>0$, let $\sum_{i=1}^n\mathbb{E}_{\mu_i}[u_i]\in[I,I+\eps)$ be such that there exists $\Delta$, satisfying $\Uppsi_{(u_i),(\Delta_j)}\geq0$ and $\inf_{\mathbb{P}\in\mathcal{P}}\mathbb{P}(\Uppsi_{(u_i),(\Delta_j)}\geq\Phi)\geq\alpha$. Define $H:=\{\Uppsi_{(u_i),(\Delta_j)}\geq\Phi\}$. By the upper semi-continuity of $\Uppsi$ and the lower semi-continuity of $\Phi$, we know that $H$ is closed. Then $H\in\mathcal{A}(\mathcal{P},\alpha)$ and $\Uppsi_{(u_i),(\Delta_j)}\geq\Phi 1_H$. Hence,
$$I+\eps>\sum_{i=1}^n\mathbb{E}_{\mu_i}[u_i]=\sup_{\mathbb{Q}\in\mathcal{M}}\mathbb{E}^{\mathbb{Q}}\left[\Uppsi_{(u_i),(\Delta_j)}\right]\geq\sup_{\mathbb{Q}\in\mathcal{M}}\mathbb{E}^{\mathbb{Q}}[\Phi 1_H]\geq J.$$
\end{proof}
\begin{remark}
In order to solve the \lq\lq inf sup\rq\rq\ problem in the first line of \eqref{eq:quantile-hed} the Neyman-Pearson Lemma needs to be generalized to the setting of \cite{Beiglbock1}, the case in which a dominating measure is absent. We leave this for future work.
\end{remark}

\bibliographystyle{siam}
\bibliography{ref}

\end{document}